\documentclass{amsart}
\usepackage{amssymb}
\usepackage{amsxtra}
\usepackage{amsmath}
\usepackage{amsfonts}
\usepackage{CJK}

\numberwithin{equation}{section}

\newtheorem{thm}{Theorem}[section]

\newtheorem{prop}[thm]{Proposition}

\newtheorem{de}[thm]{Definition}
\newtheorem{rem}[thm]{Remark}

\newcommand{\eqa}{\begin{eqnarray}}
\newcommand{\eeqa}{\end{eqnarray}}
\newcommand{\beq}{\begin{equation}}
\newcommand{\eeq}{\end{equation}}
\newcommand{\nn}{\nonumber}
\newcommand{\p}{\partial}

\def \la {\langle}
\def \ra{\rangle}

\begin{document}

\title[]
{A 2-component $\mu$-Hunter-Saxton equation}
\author[]{Dafeng Zuo}

\address[]{Department of Mathematics,University of Science and Technology of China, Hefei 230026,
P.R.China}

\email{dfzuo@ustc.edu.cn}

\subjclass[2000]{37K10, 35Q51}

\begin{abstract}
In this paper, we propose a two-component generalization of the generalized Hunter-Saxton
equation obtained in \cite{BLG2008}. We will show that this equation is a
bihamiltonian Euler equation, and also can be viewed as a bi-variational equation.
\end{abstract}

\maketitle 


\section{Introduction}

V.I.Arnold in \cite{Arn1966} suggested a general framework for the Euler equations on an arbitrary
(possibly infinite-dimensional) Lie algebra  $\mathcal{G}$. In many cases, the Euler
equations on $\mathcal{G}$ describe geodesic flows with respect to a suitable one-side invariant
Riemannian metric on the corresponding group $G$. Now it is well-known that Arnold's approach to
the Euler equation works very well for the Virasoro algebra and its extensions, see \cite{OK1987, M1998,KM2003,
Con2003,Con2006,Con2007} and references
therein.

Let $\mathcal{D}(\mathbb{S}^1)$ be a group of orientation preserving
diffeomorphisms of the circle and $G=\mathcal{D}(\mathbb{S}^1)\oplus \mathbb{R}$ be the Bott-Virasoro group.
In \cite{OK1987},  Ovsienko and Khesin showed that the KdV equation is an Euler equation, describing
a geodesic flow on $G$ with respect to a right invariant $L^2$ metric. Another
interesting example is the Camassa-Holm equation, which was originally  derived in \cite{FF1981}
as an abstract equation with a bihamiltonian structure, and independently in \cite{CH1993}
as a shallow water approximation. In \cite{M1998}, Misiolek  showed that the Camassa-Holm equation
is also an Euler equation for a geodesic flow on $G$ with respect to a right-invariant Sobolev $H^1$-metric.

In \cite{KM2003}, Khesin  and Misiolek  extended the Arnold's
approach to homogeneous spaces and provided a beautiful geometric setting for the Hunter-Saxton equation,
which  firstly appeared in \cite{HS1991} as an asymptotic equation for rotators in liquid crystals,
and its relatives. They showed that the Hunter-Saxton equation is an Euler equation describing the geodesic
flow on the homogeneous spaces of the Bott-Virasoro group $G$ modulo rotations with respect to a right
invariant homogeneous $\dot{H}^1$ metric.

Furthermore, by using extended Bott-Virasoro groups, Guha etc. \cite{Guha2000,Guha2006, Guha2008} generalized
the above results to two-component integrable systems, including several coupled KdV type systems and
 2-component peak type systems, especially 2-component Camassa-Holm equation which was introduced by Chen,
 Liu and Zhang \cite{CLZ2006} and  independently by Falqui \cite{F2006}.  Another interesting topic is
 to discuss the super or supersymmetric analogue, see \cite{OK1987, DS2001,Guha2006,JO2008, JO2009,SGD2009}
 and references therein.

Recently Khesin, Lenells and Misiolek in \cite{BLG2008} introduced a generalized Hunter-Saxton
($\mu$-HS in brief) equation lying ¡®mid-way¡¯ between the periodic Hunter-Saxton and Camassa-Holm equations,
    \beq -f_{txx} = -2\mu(f)f_x +2f_xf_{xx} +ff_{xxx},\label{eq1.1} \eeq
where $f=f(t,x)$ is a time-dependent function on the unit circle $\mathbb{S}^1 = \mathbb{R}/\mathbb{Z}$ and
$\mu(f)=\int_{\mathbb{S}^1}fdx$ denotes its mean.  This equation describes evolution of rotators in
liquid crystals with external magnetic field and self-interaction.

Let $\mathcal{D}^s(\mathbb{S}^1)$  be a group of orientation preserving
Sobolev $H^s$ diffeomorphisms of the circle. They proved that the $\mu$-HS equation \eqref{eq1.1}
 describes a geodesic flow on $\mathcal{D}^s(\mathbb{S}^1)$ with a right-invariant metric given at the
identity by the inner product
\beq \la  {f}, {g} \ra_{\mu}=\mu(f)\mu(g)+\int_{\mathbb{S}^1}f'(x)g'(x)dx.\label{eq1.2}\eeq
They also showed that \eqref{eq1.2} is bihamiltonian and admits both cusped as well as
smooth traveling-wave solutions which are natural candidates for solitons.
In this paper,  we want to generalize these to a two-component
 $\mu$-HS (2-$\mu$HS in brief) equation.  Our main object is
 the Lie algebra $\mathcal{G}=\hbox{Vect}^s(\mathbb{S}^1)
\ltimes \hbox{C}^\infty (\mathbb{S}^1)$ and its three-dimensional central extension
$\widehat{\mathcal{G}}$. Firstly, we introduce an inner product on $\widehat{\mathcal{G}}$
given by
\beq \la \hat{f},\hat{g} \ra_{\mu}=\mu(f)\mu(g)+\int_{\mathbb{S}^1}(f'(x)g'(x)+a(x)b(x))dx+
\overrightarrow{\alpha}\cdot\overrightarrow{\beta},\label{eq0.3} \eeq
where $\hat{f}=(f(x)\frac{d}{dx},a(x),\overrightarrow{\alpha})$,
$\hat{g}=(g(x)\frac{d}{dx},b(x),\overrightarrow{\beta})$
and $\overrightarrow{\alpha},\overrightarrow{\beta}\in \mathbb{R}^3$. Afterwards, we have

\begin{thm}$[$=Theorem \ref{Thm1.1}$].$ The Euler equation on $\widehat{\mathcal{G}}_{reg}^*$
with respect to \eqref{eq0.3} is a 2-$\mu$HS equation
\beq \label{eq1.3}\left\{\begin{array}{l}
 -f_{xxt}=2\mu(f)f_x-2f_xf_{xx}-ff_{xxx}+v_xv-\gamma_1f_{xxx}+\gamma_2v_{xx},\\
v_t=(vf)_x-\gamma_2f_{xx}+2\gamma_3v_x,
\end{array}\right.
\eeq
where $\gamma_j\in\mathbb{R}$, $j=1,2,3$.\end{thm}

Actually from the geometric view, if we extend the inner product \eqref{eq0.3} to a
left invariant metric on
 $\widehat{G}=\mathcal{D}^s(\mathbb{S}^1)\ltimes \hbox{C}^\infty (\mathbb{S}^1)\oplus \mathbb{R}^3$,
  we could view the  2-$\mu$HS equation \eqref{eq1.3} as a geodesic flow on  $\widehat{G}$
  with respect to this left invariant metric.  Obviously, if we choose $v=0$ and $\gamma_j=0$,
  $j=1,2,3$ and replace $t$ by $-t$, \eqref{eq1.3} reduces to \eqref{eq1.1}.
Furthermore, we show that

 \begin{thm} $[$\mbox{=Theorem \ref{Thm3.1} and \ref{Thm4.1}}$].$
 The 2-$\mu$HS equation \eqref{eq1.3} can be viewed as a bihamiltonian and bi-variational equation.
\end{thm}

This paper is organized as follows. In section 2, we calculate the Euler equation
on $\widehat{\mathcal{G}}_{reg}^*$. In section 3,  we study the Hamiltonian nature and the Lax pair of the
2-$\mu$HS equation  \eqref{eq1.3}. Section 4 is devoted to discuss the variational nature of \eqref{eq1.3}.
In the last section we describe the interrelation between bihamiltonian natures and bi-variational natures.\\

\noindent {\bf Acknowledgement.} The author would like to thank
Prof. Khesin Boris and Prof. Partha Guha for
references \cite{OK1987} and \cite{Guha2008}, respectively, and
the anonymous referee for several useful suggestions.
 This work is partially supported by the Fundamental Research Funds for the Central
 Universities and NSFC(10971209,10871184).


\section{Eulerian nature of the 2-$\mu$HS equation}
Let $\mathcal{D}^s(\mathbb{S}^1)$  be a group of orientation preserving Sobolev $H^s$ diffeomorphisms of
the circle and let $T_{id}\mathcal{D}^s(\mathbb{S}^1)$ be the corresponding Lie algebra of vector fields, denoted by
$\hbox{Vect}^s(\mathbb{S}^1)=\{ f(x)\frac{d}{dx}|f(x)\in H^s(\mathbb{S}^1)\}$.

The main objects in our paper will be the group  $\mathcal{D}^s(\mathbb{S}^1)\ltimes
\hbox{C}^\infty (\mathbb{S}^1)$, its Lie algebra $\mathcal{G}=\hbox{Vect}^s(\mathbb{S}^1)
\ltimes \hbox{C}^\infty (\mathbb{S}^1)$ with the
Lie bracket given by
\beq [(f(x)\frac{d}{dx}, a(x)),(g(x)\frac{d}{dx}, b(x))]
=\left( (f(x)g'(x)-f'(x)g(x))\frac{d}{dx}, f(x)b'(x)-a'(x)g(x)\right),\nn \eeq
 and their central extensions.  It is well known in \cite{DM1970, AKP1988} that the algebra $\mathcal{G}$ has a
 three dimensional central extension given by the following nontrivial cocycles
\eqa
 && \omega_1\left( (f(x)\frac{d}{dx}, a(x)),(g(x)\frac{d}{dx}, b(x)) \right)=\int_{\mathbb{S}^1}f'(x)g''(x)dx, \\
  && \omega_2\left( (f(x)\frac{d}{dx}, a(x)),(g(x)\frac{d}{dx}, b(x)) \right)
  =\int_{\mathbb{S}^1}[f''(x)b(x)-g''(x)a(x)]dx,\nn \\
 && \omega_3\left( (f(x)\frac{d}{dx}, a(x)),(g(x)\frac{d}{dx}, b(x)) \right)=2\int_{\mathbb{S}^1}a(x)b''(x)dx, \nn
 \eeqa
 where $f(x)$, $g(x)\in H^s(\mathbb{S}^1)$ and  $a(x)$, $b(x) \in \hbox{C}^\infty (\mathbb{S}^1)$.
 Notice that the first cocycle
 $\omega_1$ is the well-known Gelfand-Fuchs cocycle \cite{GF1968,F1984}. The Virasoro algebra
 $Vir=\hbox{Vect}^s(\mathbb{S}^1)\oplus\mathbb{R}$
 is the unique non-trivial central extension of $\hbox{Vect}^s(\mathbb{S}^1)$ via the Gelfand-Fuchs cocycle $\omega_1$.
 Sometimes we would like to use the modified Gelfand-Fuchs cocycle
\beq \tilde{\omega}_1\left( (f(x)\frac{d}{dx}, a(x)),(g(x)\frac{d}{dx}, b(x)) \right)=
\int_{\mathbb{S}^1}(c_1f'(x)g''(x)+c_2 f'(x)g(x))dx,\eeq
which is cohomologeous to the Gelfand-Fuchs cocycle $\omega_1$, where $c_1$, $c_2\in \mathbb{R}$.

\begin{de}The algebra $\widehat{\mathcal{G}}$ is an extension of $\mathcal{G}$ defined by
\beq \widehat{\mathcal{G}}=\hbox{Vect}^s(\mathbb{S}^1)\ltimes \hbox{C}^\infty (\mathbb{S}^1)\oplus\mathbb{R}^3\eeq
with the commutation relation
\beq [\hat{f},\hat{g}]
=\left( (fg'-f'g)\frac{d}{dx}, fb'-a'g,\overrightarrow{\omega}\right),\eeq
where $\hat{f}=(f(x)\frac{d}{dx},a(x),\overrightarrow{\alpha})$,
$\hat{g}=(g(x)\frac{d}{dx},b(x),\overrightarrow{\beta})$
and $\overrightarrow{\alpha},\overrightarrow{\beta}\in \mathbb{R}^3$ and
$\overrightarrow{\omega}=(\omega_1,\omega_2,\omega_3)\in \mathbb{R}^3$.
\end{de}

Let
\beq \widehat{\mathcal{G}}_{reg}^*=\hbox{C}^\infty (\mathbb{S}^1)\oplus\hbox{C}^\infty
(\mathbb{S}^1)\oplus\mathbb{R}^3\eeq
denote the regular part of the dual space $\widehat{\mathcal{G}}^*$ to the Lie algebra
$\widehat{\mathcal{G}}$,
under the pairing
\beq \la \hat{u}, \hat{f}\ra^*=\int_{\mathbb{S}^1} (u(x)f(x)+a(x)v(x))dx+
\overrightarrow{\alpha}\cdot\overrightarrow{\gamma},\eeq
where $\hat{u}=(u(x)(dx)^2,v(x),\overrightarrow{\gamma})\in  \widehat{\mathcal{G}}^*$.
Of particular interest are the coadjoint orbits in  $\widehat{\mathcal{G}}_{reg}^*$.

On $\widehat{\mathcal{G}}$, let us introduce an inner product
\beq \la \hat{f},\hat{g} \ra_{\mu}=\mu(f)\mu(g)+\int_{\mathbb{S}^1}(f'(x)g'(x)+a(x)b(x))dx+
\overrightarrow{\alpha}\cdot\overrightarrow{\beta}.\label{eq2.7} \eeq
 A direct computation gives
\beq
\la \hat{f},\hat{g} \ra_{\mu}=\la \hat{f}, (\Lambda(g)(dx)^2, b(x),\overrightarrow{\beta})\ra^*,
\quad \Lambda(g)=\mu(g)-g''(x),\nn
\eeq
which induces an inertia operator
 $\mathcal{A}:\widehat{\mathcal{G}}\longrightarrow \widehat{\mathcal{G}}^*_{reg}$
given by
\beq \mathcal{A}(\hat{g})=(\Lambda(g)(dx)^2, b(x),\overrightarrow{\beta}).\eeq

\begin{thm}\label{Thm1.1}The 2-$\mu$HS equation \eqref{eq1.3} is an Euler equation on $\widehat{\mathcal{G}}_{reg}^*$
with respect to the inner product \eqref{eq2.7}. \end{thm}

\begin{proof}By definition,
\eqa &&\la ad^*_{\hat{f}}(\hat{u}),\hat{g} \ra^*=-\la \hat{u},[\hat{f},\hat{g}]\ra^* \qquad
\mbox{by using integration by parts} \nn \\
&&\quad=\la \left((2uf_x+u_xf+a_xv-\alpha_1f_{xxx}+\alpha_2a_{xx})(dx)^2,
(v f)_x-\alpha_2f_{xx}+2\alpha_3a_x, 0\right),\hat{g}\ra^*.\nn
 \eeqa
This gives
\beq ad^*_{\hat{f}}(\hat{u})=((2uf_x+u_xf+a_xv-\alpha_1f_{xxx}+\alpha_2a_{xx})(dx)^2,
(v f)_x-\alpha_2f_{xx}+2\alpha_3a_x, 0).\nn\eeq

By definition in \cite{KM2003},  the Euler equation on $\widehat{\mathcal{G}}_{reg}^*$ is given by
\beq \frac{d\hat{u}}{dt}=ad^*_{\mathcal{A}^{-1} \hat{u}} \hat{u}\label{eq2.9}\eeq
as an evolution of a point $\hat{u}\in \widehat{\mathcal{G}}_{reg}^*$. That is to say,
the Euler equation on $\widehat{\mathcal{G}}_{reg}^*$ is
\eqa\
&& u_t=2uf_x+u_xf+v_xv-\gamma_1f_{xxx}+\gamma_2v_{xx},\nn\\
&& v_t=(vf)_x-\gamma_2f_{xx}+2\gamma_3v_x,\nn
\eeqa
 where $u(x,t)=\Lambda(f(x,t))=\mu(f)-f_{xx}$. By integrating both sides of
 this equation over the circle and using periodicity, we obtain
  $$\mu(f_t)=\mu(f)_t=0.$$
  This yields that
\eqa
&& -f_{xxt}=2\mu(f)f_x-2f_xf_{xx}-ff_{xxx}+v_xv-\gamma_1f_{xxx}+\gamma_2v_{xx},\nn\\
&& \quad v_t=(vf)_x-\gamma_2f_{xx}+2\gamma_3v_x,\nn
\eeqa
which is the 2-$\mu$HS equation \eqref{eq1.3}. \end{proof}

\begin{rem}If we replace the Gelfand-Fuchs cocycle $\omega_1$ by the modified cocycle $\tilde{\omega}_1$,
the Euler equation
$\widehat{\mathcal{G}}_{reg}^*$ is of the form
\eqa
&& -f_{xxt}=2\mu(f)f_x-2f_xf_{xx}-ff_{xxx}+v_xv-\gamma_1c_1f_{xxx}+\gamma_2v_{xx}+\gamma_1c_2f_x,\nn\\
&&\quad v_t=(vf)_x-\gamma_2f_{xx}+2\gamma_3v_x.\nn
\eeqa
\end{rem}

\section{Hamiltonian nature of the 2-$\mu$HS equation}
In this section, we want to study the Hamiltonian nature of the 2-$\mu$HS equation \eqref{eq1.3}
and its geometric meaning.
We will show that

\begin{thm}\label{Thm3.1}The 2-$\mu$HS equation \eqref{eq1.3} is bihamiltonian.\end{thm}
\begin{proof}
Let us define  $u(x,t)=\Lambda(f)=\mu(f)-f_{xx}$ and
\beq H_1=\frac{1}{2}\int_{\mathbb{S}^1}(u f+v^2)dx \label{eq3.1} \eeq
and
\beq H_2=\int_{\mathbb{S}^1} (\mu(f)f^2+\frac{1}{2}ff_x^2+\frac{1}{2}fv^2-
\gamma_2vf_x+\gamma_3v^2-\frac{\gamma_1}{2}ff_{xx})dx. \label{eq3.2}\eeq
It is easy to check  that the 2-$\mu$HS equation can be written as
\beq \left(\begin{array}{l} u\\v\end{array}\right)_t=\mathcal{J}_1
\left(\begin{array}{l} \frac{\delta H_2}{\delta u}\\
\frac{\delta H_2}{\delta v}\end{array}\right)=\mathcal{J}_2
\left(\begin{array}{l} \frac{\delta H_1}{\delta u}\\
\frac{\delta H_1}{\delta v}\end{array}\right),\eeq
where the Hamiltonian operators are
\beq
\mathcal{J}_1=\left(\begin{array}{cc}
\p_x \Lambda & 0\\
0& \p_x
\end{array}\right), \quad \mathcal{J}_2=\left(\begin{array}{cc}
 u\p_x+\p_xu-\gamma_1\p_x^3&v\p_x+\gamma_2\p_x^2\\
\p_xv-\gamma_2\p^2_x& 2\gamma_3\p_x
\end{array}\right).
\eeq
By a direct and lengthy calculation we could show that Hamiltonian operators $\mathcal{J}_1$ and $\mathcal{J}_2$
are compatible. \end{proof}

Next we want to explain the geometric meaning of the bihamiltonian structures of the 2-$\mu$HS
equation \eqref{eq1.3}. Let $F_i:\widehat{\mathcal{G}}_{reg}^* \to \mathbb{R}$, $i=1,2$,
be two arbitrary smooth functionals.
It is well-known that the dual space $\widehat{\mathcal{G}}_{reg}^*$ carries the canonical Lie-Poisson bracket
\beq \{F_1,F_2\}_2(\hat{u})=\la \hat{u}, [\frac{\delta F_1}{\delta \hat{u}},
\frac{\delta F_2}{\delta \hat{u}}] \ra^*,\label{eq3.5} \eeq
where $\hat{u}=(u(x,t)(dx)^2,v(x,t),\vec{\gamma}) \in \widehat{\mathcal{G}}_{reg}^*$
and $ \frac{\delta F_i}{\delta \hat{u}}
=(\frac{\delta F_i}{\delta u},\frac{\delta F_i}{\delta v}, \frac{\delta F_i}{\delta \vec{\gamma}})
\in \widehat{\mathcal{G}},i=1,2$.
By definition of the Euler
equation \eqref{eq2.9}, we know that the Lie-Poisson structure \eqref{eq3.5} is exactly the second
Poisson bracket, induced by $\mathcal{J}_2$,
of the 2-$\mu$HS equation \eqref{eq1.3}.

To explain the first Hamiltonian structure, in the following we will use  the ``frozen  Lie-Poisson''
 method introduced in \cite{KM2003}. Let us define
a frozen (or constant) Poisson bracket
\beq \{F_1,F_2\}_1(\hat{u})=\la \hat{u}_0, [\frac{\delta F_1}{\delta \hat{u}},
\frac{\delta F_2}{\delta \hat{u}}] \ra^*,\label{eq3.6} \eeq
where $\hat{u}_0=(u_0(dx)^2,v_0,\vec{\gamma}_0) \in \widehat{\mathcal{G}}_{reg}^*$.
The corresponding Hamiltonian equation for
any functional $F:\widehat{\mathcal{G}}_{reg}^* \to \mathbb{R}$ reads
  \beq \frac{d\hat{u}}{dt}=ad^*_{\frac{\delta F}{\delta \hat{u}}} \hat{u}_0 \label{eq3.7}\eeq
which gives
\eqa\
&& u_t=2u_0(\frac{\delta F}{\delta {u}})_x+(\frac{\delta F}{\delta {v}})_xv_0
-\gamma_1^0 (\frac{\delta F}{\delta {u}})_{xxx}+\gamma_2^0(\frac{\delta F}{\delta {v}})_{xx},\nn\\
&& v_t=(v_0\frac{\delta F}{\delta {u}})_x-\gamma_2^0(\frac{\delta F}{\delta {u}})_{xx}
+2\gamma_3^0(\frac{\delta F}{\delta {v}})_x,\label{eq3.8}\\
&&\vec{\gamma}_{0,t}=0.\nn
\eeqa
 Let us take the Hamiltonian functional $F$ to be
 \beq  H_2=\int_{\mathbb{S}^1} (\mu(f)f^2+\frac{1}{2}ff_x^2+\frac{1}{2}fv^2-
\gamma_2vf_x+\gamma_3v^2-\frac{\gamma_1}{2}ff_{xx})dx \eeq
and set $u(x,t)=\Lambda(f(x,t))=\mu(f)-f_{xx}$.  Then we have
\eqa && \frac{\delta F}{\delta {u}}=\Lambda^{-1}(\mu(f^2)+2f\mu(f)-\frac{1}{2}f_x^2-ff_{xx}
-\gamma_1f_{xx}+\gamma_2v_x),\nn\\
&& \frac{\delta F}{\delta {v}}=vf-\gamma_2f_x+2\gamma_3v.\label{eq3.10}\eeqa
Let us choose a fixed point
$$\hat{u}_0=(u_0,v_0,\vec{\gamma}_0)=(0,0, (1,0,\frac{1}{2})).$$
Observe that $\p_x^3 \Lambda^{-1}=-\p_x$. By substituting \eqref{eq3.10} into \eqref{eq3.8}, we
obtain the 2-$\mu$HS equation \eqref{eq1.3}.
According to the {\bf Proposition 5.3} in \cite{KM2003}, $\{~,~\}_1$ and $\{~,~\}_2$ are
compatible for every freezing point $\hat{u}_0$. Consequently we have

\begin{thm}The 2-$\mu$HS equation \eqref{eq1.3} is Hamiltonian with respect to
two compatible Poisson structures \eqref{eq3.5} and \eqref{eq3.6} on $\widehat{\mathcal{G}}_{reg}^*$,
where the first bracket is frozen at the point $\hat{u}_0=(u_0,v_0,\vec{\gamma}_0)=(0,0, (1,0,\frac{1}{2}))$.\end{thm}

 Let us point out that the constant bracket depends on the choice of the freezing
point $\hat{u}_0$, while the Lie-Poisson bracket is only determined by the Lie algebra structure.

To this end we want to derive a Lax pair of 2-$\mu$HS equation \eqref{eq1.3} with $\overrightarrow{\gamma}=0$, i.e.,
\beq \label{eq3.11}
 -f_{xxt}=2\mu(f)f_x-2f_xf_{xx}-ff_{xxx}+v_xv,\quad
v_t=(vf)_x.
\eeq
Motivated by the Lax pair of the two-component Camassa-Holm equation in \cite{CLZ2006}, we
could assume that the Lax pair of \eqref{eq3.11} has the following form
\beq \Psi_x=U\Psi, \quad \Psi_t=V\Psi \label{eq3.12}\eeq
with
\beq U=\left(\begin{array}{cc}
0&1\\
\lambda \Lambda(f)-\lambda^2 v^2&0
\end{array}\right)
\quad \mbox{and}\quad
V=\left(\begin{array}{cc}
p & r\\
q & -p
\end{array}\right),\nn \eeq
where $\lambda$ is a spectral parameter. The compatibility condition $$U_t-V_x+UV-VU=0$$
in componentwise form reads
\eqa
 && p=-\frac{r_x}{2}, \quad q=p_x+r(\lambda \Lambda(f)-\lambda^2 v^2),\nn\\
&& 2\lambda^2vv_t+\lambda f_{xxt}+q_x-2p(\lambda \Lambda(f)-\lambda^2 v^2)=0.\nn
\eeqa
By choosing $r=f-\frac{1}{2\lambda}$, we have
$$p=-\frac{f_x}{2},\quad q=-\frac{f_{xx}}{2}+(f-\frac{1}{2\lambda})(\lambda \Lambda(f)-\lambda^2 v^2)$$
and
$$f_{xxt}+2\mu(f)f_x-2f_xf_{xx}-ff_{xxx}+v_xv+2\lambda v(
v_t-(vf)_x)=0$$
which yields the system  \eqref{eq3.11}. Let us write
$\Psi=\left(\begin{array}{c}\psi
\\ \psi_x\end{array}\right),$
we have

\begin{prop} The system \eqref{eq3.11} has a Lax pair given by
 \beq
 \psi_{xx}=(\lambda \Lambda(f)-\lambda^2 v^2)\psi, \quad
 \psi_t=(f-\frac{1}{2\lambda})\psi_x-\frac{1}{2}f_x\psi, \nn
 \eeq where $\lambda \in \mathbb{C}-\{0\}$ is a spectral parameter.
 \end{prop}

\section{Variational nature of the 2-$\mu$HS equation}
In \cite{BLG2008}, they have shown that the $\mu$-HS equation \eqref{eq1.1}
can be obtained from two distinct variational principles.
In this section we will show that the 2-$\mu$HS equation \eqref{eq1.3} also
arises as the equation
    $$\delta \mathcal{S}=0$$
for the action functional
$$ \mathcal{S}=\int(\int\mathcal{L}dx)dt$$
with two different densities $\mathcal{L}$. That is to say,

\begin{thm}\label{Thm4.1} The 2-$\mu$HS equation \eqref{eq1.3} satisfies two different
variational principles.\end{thm}

\begin{proof} Motivated by the Lagrangian densities for the
$\mu$-HS equation \eqref{eq1.1} in \cite{BLG2008}, by some conjectural computations
we find two generalized Lagrangian densities for the 2-$\mu$HS equation \eqref{eq1.3}.
More precisely,\\

\noindent{\bf Case I}. Let us consider the first Lagrangian density
\beq \mathcal{L}_1=\frac{1}{2}f_x^2+\frac{1}{2}\mu(f)f+\frac{1}{2}v^2-vz_x+w(fz_x-z_t+
\tilde{\gamma}_3v)+\gamma_2 w_x f-2\gamma_1 f,\eeq
where $\tilde{\gamma}_3=\gamma_3-\frac{1}{2}\gamma_1$.
Varying the corresponding action with respect to $f$, $v$, $w$ and $z$ respectively, we get
\beq\label{eq2.16}
\begin{array}{l}
f_{xx}=\mu(f)+wz_x+\gamma_2 w_x-2\gamma_1,\\
z_x=v+\tilde{\gamma}_3 w,\\
z_t=fz_x+\tilde{\gamma}_3v-\gamma_2f_x,\\
w_t=(wf)_x-v_x. \end{array}
\eeq
By using \eqref{eq2.16}, we have
\eqa
&&v_t=z_{xt}-\gamma_3 w_t=[f(v+\tilde{\gamma}_3 w)+\tilde{\gamma}_3v-\gamma_2f_x]_x-\tilde{\gamma}_3 ((wf)_x-v_x),\nn\\
&&\quad=(vf)_x-\gamma_2f_{xx}+(2 {\gamma}_3-\gamma_1)v_{x},\label{eq2.17}
\eeqa
and
\eqa
&&-f_{xxt}+f_xf_{xx}+ff_{xxx}\nn\\
&&\qquad =-(\mu(f)+wz_x+\gamma_2 w_x)_t+f_x(\mu(f)+wz_x+\gamma_2 w_x-2\gamma_1)+f(\mu(f)+wz_x+\gamma_2 w_x)_x\nn\\
&&\qquad=-w_{t}z_x-wz_{xt}+\gamma w_{xt}+f_xwz_x+fw_xz_x+fwz_{xx}+\gamma_2fw_{xx}+2\gamma_1f_x\nn\\
&&\qquad=vv_x+2\mu(f)f_x+\gamma_2v_{xx}-2\gamma_1f_x.\label{eq2.18}
\eeqa
Notice that if we replace $f$ by $f+\gamma_1$ in the system \eqref{eq2.17} and \eqref{eq2.18}, this
gives the 2-$\mu$HS equation \eqref{eq1.3}.\\

\noindent {\bf Case II}. The second variational representation can be obtained from the Lagrangian density
\beq \mathcal{L}_2=-f_xf_t+2\mu(f)f^2+ff_x^2+f\phi_x^2-\gamma_1ff_{xx}-2\gamma_2\phi_xf_x+2\gamma_3\phi_x^2
-\phi_x\phi_t.\eeq
The variational principle $\delta \mathcal{S}=0$ gives the Euler-Lagrange equation
\beq \label{eq2.15}\begin{array}{l}
 -f_{xt}=2\mu(f)f+\mu(f^2)-\frac{1}{2}f_x^2-ff_{xx}+
\frac{1}{2}\phi_x^2-\gamma_1f_{x}+\gamma_2\phi_{xx},\\
~~~\phi_{xt}=(f\phi_x)_x-\gamma_2f_{xx}+2\gamma_3\phi_{xx}.
\end{array}
\eeq
If we set $\phi_x=v$ and take the $x$-derivative of the first term in \eqref{eq2.15}, this yields
the 2-$\mu$HS equation \eqref{eq1.3}.\end{proof}

\section{Relation between Hamiltonian nature and Variational nature} Recall that we have shown that
the 2-$\mu$HS equation \eqref{eq1.3} is bihamiltonian  and has two different variational principles.
In the last section we want to study the relation between Hamiltonian natures and bi-variational
principles and prove that

\begin{thm} The two variational formulations for the 2-$\mu$HS equation \eqref{eq1.3} formally correspond
to the two Hamiltonian formulations of this equation with Hamiltonian functionals
$H_1$ and $H_2$. \end{thm}

\begin{proof} The action is related to the Lagrangian by  $\mathcal{S}=\int(\int\mathcal{L}dx)dt$. The
first variational principle has the Lagrangian density,
\beq \mathcal{L}_1=\frac{1}{2}f_x^2+\frac{1}{2}\mu(f)f+\frac{1}{2}v^2-vz_x+w(fz_x-z_t+
\tilde{\gamma}_3v)+\gamma_2 w_x f-2\gamma_1f.\nn\eeq
The momenta conjugate to the velocities $f_t$, $v_t$, $z_t$ and $w_t$, respectively, are
\beq \frac{\p \mathcal{L}_1}{\p f_t}=0,\quad \frac{\p \mathcal{L}_1}{\p w_t}=0,
\quad \frac{\p \mathcal{L}_1}{\p z_t}=-w,\quad \frac{\p \mathcal{L}_1}{\p w_t}=0.\nn\eeq
Consequently, the Hamiltonian density is
\eqa
&&\mathcal{H}=-z_tw-\mathcal{L}_1\nn\\
&&\qquad =-\frac{1}{2}f_x^2-\frac{1}{2}\mu(f)f-\frac{1}{2}v^2+vz_x-w(fz_x+
\tilde{\gamma}_3v)-\gamma_2 w_x f+2\gamma_1\nn\\
&&\qquad=\frac{1}{2}\mu(f)f-\frac{1}{2}f_x^2+\frac{1}{2}v^2-ff_{xx}, \quad \mbox{by using \eqref{eq2.16}}.\nn
\eeqa
Therefore, the Hamiltonian is
\eqa &&H=\int \mathcal{H}dx=\int (\frac{1}{2}\mu(f)f-\frac{1}{2}f_x^2+\frac{1}{2}v^2-ff_{xx}) dx\nn\\
&&\quad=\frac{1}{2}\int(\mu(f)f-ff_{xx}+v^2)dx,\nn \eeqa
which is exactly $H_1$ defined in \eqref{eq3.1}.

In the second principle the Lagrangian density is
\beq \mathcal{L}_2=-f_xf_t+2\mu(f)f^2+ff_x^2+f\phi_x^2-\gamma_1ff_{xx}-2\gamma_2\phi_xf_x+2\gamma_3\phi_x^2
-\phi_x\phi_t.\nn\eeq
The momenta conjugate to the velocities $f_t$ and $\phi_t$, respectively, are
\beq \frac{\p \mathcal{L}_2}{\p f_t}=-f_x,\quad \frac{\p \mathcal{L}_2}{\p \phi_t}=-\phi_x.\nn\eeq
Consequently, the Hamiltonian density is
\eqa
&&\mathcal{H}=-f_xf_t-\phi_x\phi_t-\mathcal{L}_2\nn\\
&&\qquad =-2\mu(f)f^2-ff_x^2-f\phi_x^2+\gamma_1ff_{xx}+2\gamma_2\phi_xf_x-2\gamma_3\phi_x^2\nn
\eeqa
Now let us set $\phi_x=v$ and so
\beq H=\int (-2\mu(f)f^2-ff_x^2-fv^2+-\gamma_1ff_{xx}+2\gamma_2vf_x-2\gamma_3v^2)dx=-\dfrac{H_2}{2},\nn\eeq
where $H_2$ is defined in \eqref{eq3.2}.\end{proof}


\end{document}